\documentclass{ifacconf}

\usepackage{graphicx}  
\usepackage{natbib}        

\usepackage{amsmath}
\usepackage{amssymb}
\usepackage{algorithm}
\usepackage{algpseudocode}
\usepackage{booktabs}
\usepackage{subcaption}
\usepackage{siunitx}

\usepackage{stmaryrd}
\usepackage{xcolor}

\newcommand{\sg}[1]{{\color{black} #1}}

\usepackage{microtype}
\newtheorem{assumption}{Assumption}

\newtheorem{proposition}{Proposition}
\newtheorem{remark}{Remark}
\newtheorem{lemma}{Lemma}

\begin{document}
\begin{frontmatter}

\title{{Data-Driven Network  LQG Mean Field Games with Heterogeneous Populations via Integral Reinforcement Learning}
\thanksref{footnoteinfo}}

\thanks[footnoteinfo]{This work is supported in part by  NSERC Discovery Grant  RGPIN-2024-06612 and FRQNT  grant No. 359699 (https://doi.org/10.69777/359699). J.Z. was supported by the NSERC Undergraduate Student Research Award.
}

\author[First]{Jean Zhu} 
\author[Second]{Shuang Gao} 

\address[First]{Department of Electrical Engineering, Polytechnique Montréal \& GERAD, Montreal, Canada (email: jean.zhu@etud.polymtl.ca)}
\address[Second]{Department of Electrical Engineering, Polytechnique Montréal \& GERAD, Montreal, Canada (email: shuang.gao@polymtl.ca)}

\allowdisplaybreaks
\begin{abstract}                
This paper establishes a data-driven solution for infinite-horizon linear quadratic Gaussian Mean Field Games with network-coupled heterogeneous agent populations where the dynamics of the agents are unknown. The solution technique relies on  Integral Reinforcement Learning  and Kleinman's iteration for solving algebraic Riccati equations (ARE). The resulting algorithm uses trajectory data to generate network-coupled MFG strategies for agents and does not require  parameters of agents' dynamics. Under technical conditions on the persistency of excitation and on the existence of unique stabilizing solution to the corresponding AREs, the learned network-coupled MFG strategies are shown to converge to their true values. 
\end{abstract}
\end{frontmatter}

\section{Introduction}
Large-scale systems composed of heterogeneous agent populations, such as renewable energy grids with different types of renewable sources and various types of users, and autonomous vehicle networks with different types of vehicles, {naturally arise in the transition to sustainable energy systems.}  
Such systems motivate the modelling and the strategy design in this paper for heterogeneous populations of strategically competitive agents. 
%
%
To achieve tractable strategy design for large populations of competitive agents,  Mean Field Game (MFG) theory  was proposed by \cite{Huang:2006, Huang:2007} and independently by \cite{Lasry:2006, Lasry:2007}.   
In the linear-quadratic-Gaussian (LQG) case, the MFG solution involves solving coupled Riccati equations in \citep{Huang:2007,Huang:2010}.
{MFG problems with multiple classes have been investigated in  \citep{Huang:2006, Huang:2007, Huang:2010} and MFGs with network interactions in e.g. \citep{Huang:2010locality,Gao:2023}.}

For MFGs  where the dynamics of agents are unknown,  several data-driven solutions have been established. In continuous-time settings, {adaptive control techniques with system identification  have been applied to MFGs with heterogeneous populations with mean field cost couplings by \cite{Kizilkale:2012}},  and  integral reinforcement learning (IRL) has been used by \cite{Xu:23, Xu:25,Li:2025} to generate the data-driven strategies for LQG mean field game problems  without system identification,
%
where a homogeneous population of agents with mean field coupling is assumed.  
For discrete-time MFGs, standard reinforcement learning techniques  have been employed in \citep{Jayakumar:2019,Guo:2019,Fu:2019, Zaman:2020, Zaman:2023,Angiuli:2022}.

 IRL  was developed to generate data-driven optimal control solutions for continuous-time systems with unknown dynamics (e.g.  partially unknown nonlinear dynamics \citep{Vrabie:2009},  completely unknown linear dynamics \citep{Jiang:2012} and stochastic dynamics \citep{Li:2022}). Such a data-driven technique also applies to multi-agent systems with unknown dynamics in \citep{Vamvoudakis:2011} and cases with heterogeneous agents in \citep{Modares:2016}.
 In the context of mean field control with agent populations, IRL was used for  continuous-time problems in \citep{Xu:2025b} whereas standard reinforcement learning has been employed in \citep{Jayakumar:2019,Carmona:2019a,Angiuli:2022} for discrete-time problems.


\textbf{Contribution}: This paper establishes (a) continous time infinite-horizon LQG-MFG strategies for
multi-class heterogeneous populations with inter-class network couplings and  (b) an IRL algorithm that learns these  strategies from trajectory data without system identification. Existing IRL-based LQG-MFG methods  \citep{Xu:23, Xu:25, Li:2025} assume a homogeneous population with mean field coupling, and rely on solving two algebraic Riccati equations (ARE) based on data.  In contrast, we establish that    the network interaction together with the heterogeneity of agent classes leads to a new set of AREs, one per class, and one capturing cross-class interactions. Furthermore, we show that this new  set of AREs can be learned simultaneously  from trajectory data without knowing the parameters of the underlying agent dynamics.


\textbf{Notation {and Definition}}:
$\mathbb{R}$ and $\mathbb{N}$ denote the set of real and nonzero natural numbers respectively. For a random variable $x$, $\bar{x}$ denotes its expectation $\bar{x} = \mathbb{E}\left[x\right]$. For a matrix $A$, $A^\top$ denotes its transpose. For any $n,m \in \mathbb{N}$, let $I_n \in \mathbb{R}^{n \times n}$ denote the identity matrix, $\mathbf{0}_{n\times m} \in \mathbb{R}^{n \times m}$ the zero matrix, and $\mathbf{1}_{n \times m} \in \mathbb{R}^{n\times m}$ the matrix of ones.  For a matrix $A \in \mathbb{R}^{m\times n}$, $\text{vec}(A)$ corresponds to the $mn$-dimensional vector formed by stacking the columns of $A$ on top of each other. $\otimes $ denotes the Kronecker product between two matrices. $A \succeq 0$ (resp. $A \succ 0$) denotes that matrix $A$ is positive semidefinite (resp. positive definite). $\text{diag}(M_1, \ldots, M_N)$  denotes the  matrix with diagonal blocks $M_1, \cdots, M_N$  and zero elsewhere. 
For $K, n_k, m_l, N, M \in \mathbb{N}$, $\mathcal{M} = \llbracket M_{kl}\rrbracket \in \mathbb{R}^{N \times M}$, where $N = \sum_{k=1}^Kn_k$ and $M = \sum_{l=1}^Km_l$, denotes a block matrix with blocks $M_{kl} \in \mathbb{R}^{n_k \times m_l}$, with $k, l \in \{1,...,K\}$.

  Consider {a matrix $A\in \mathbb{R}^{k \times k}$. $A$ is called \emph{strong ($k_1, k_2$) c-splitting} if, of its $k=k_1+k_2$ eigenvalues, it has $k_1$ eigenvalues in the open left half plane, and $k_2$ in the open right half plane, with none on the imaginary axis. 
 An $l$-dimensional subspace (with $l\leq k$) of  $A$  denoted by $\mathcal{V}$ of $\mathbb{R}^k$ is  an \emph{invariant subspace}  if $A\mathcal{V}\subset \mathcal{V}$; in this case, $AV = V A_0$  for some $A_0 \in \mathbb{R}^{l\times l}$ where $V \in \mathbb{R}^{k\times l}$ and $\text{span}{(V)} = \mathcal{V}$. If $A_0$ is Hurwitz, $\mathcal{V}$ is called a \emph{stable invariant subspace}.
    An $l$-dimensional  subspace (with $l\leq k$) of  $A\in \mathbb{R}^{k \times k}$ is called a \emph{graph subspace} if it's spanned by the columns of a $k\times l$ matrix whose first $l$ rows form an invertible submatrix (see \cite{Huang:2019}). 

\section{Problem formulation}
Consider a population of agents  grouped into $\mathcal{K}$ classes, where each class $k \in \{ 1, \ldots, \mathcal{K\}}$ contains a large number of identical agents. Let $c_k$ denote the number of agents in class $k$. An agent $a_{k,i}$ of class $k$ has dynamics given by 
\begin{equation}
    dx_{k,i}(t) = (A_{k}x_{k,i}(t) + B_ku_{k,i}(t))dt + D_kdw_{k,i}(t)
\end{equation}
where $x_{k,i}(t) \in \mathbb{R}^{n_k}$ and $u_{k,i}(t) \in \mathbb{R}^{m_k}$ are the states and control inputs respectively, and $w_{k,i}(t) \in \mathbb{R}^{d_k}$ is a standard Wiener process. $A_k$, $B_k$ and $D_k$ are system matrices for class $k$, which are all assumed to be unknown. For simplicity of presentation, we omit the time index hereafter.
An agent $a_{k,i}$ seeks to find the control $u_{k,i}$ that minimizes the class-specific discounted cost functional 
\begin{equation}
J_k\sg{(u_{k,i})} = \mathbb{E} \int_0^\infty e^{-\rho t} \left( e_{k,i}^\top Q_k e_{k,i} + u_{k,i}^\top R_k u_{k,i} \right) dt
\end{equation}
where $e_{k,i} = x_{k,i} - \sum_{m=1}^{\mathcal{K}}H_{km}\tilde{x}_m$ is the tracking error with $\tilde{x}_m = \frac{1}{c_m}\sum_{i=1}^{c_m} x_{m,i}$, $Q_k=Q_k^\top \succeq 0$ and $R_k = R_k^\top \succ 0$ the class-specific state cost and control cost matrices respectively, $\rho > 0$ is the discount factor, and $H_{km} \in \mathbb{R}^{n_k \times n_m}$ is a class-mean coupling matrix representing how the states of an agent of class $k$ depends on the means of the states of class $m$. The resulting $\mathcal{H} = \llbracket H_{km}\rrbracket \in \mathbb{R}^{N \times N}$ with $N = \sum_{k=1}^{\mathcal{K}}n_k$ is the \sg{network} coupling matrix.

\begin{remark}
In the special case where the agents have homogeneous state dimensions $n_k = n$ for all $k$ and are to track the global mean field $\Bar{x}(t)$, defined as the convex combination of the mean of all classes
\begin{equation}
    \Bar{x}(t)= \sum_{k=1}^{\mathcal{K}}\pi_k \Bar{x}_k(t)\text{,}\quad\pi_k > 0\text{,}\quad\sum_{k=1}^{\mathcal{K}}\pi_k = 1
\end{equation}
where $\Bar{x}_k(t)$ is the mean of all agents belonging to class $k$ and $\pi_k$ are scalar weights, then the coupling matrices $H_{km}$ are to be chosen such that $H_{km} = \pi_mI_n$ for all $k, m \in \{1,\ldots,\mathcal{K}\}$.
\end{remark}

\begin{assumption} \label{assum:A_B_stabi_A_Q_Obs}
    The pair ($A_k-\frac{1}{2}\rho I_{n_k}, B_k$) is stabilizable and the pair ($A_k-\frac{1}{2}\rho I_{n_k}, Q_k^{\frac{1}{2}}$) is observable for all class $k\in \{1,,,..\mathcal{K}\}$\sg{.}
\end{assumption}
The following block diagonal matrices are then defined.
\begin{equation*}
\begin{aligned}
    \mathcal{A} &= \text{diag}(A_1, 
    \ldots, A_\mathcal{K}), & \mathcal{B} &= \text{diag}(B_1, \ldots, B_\mathcal{K}),  \\
    \mathcal{Q} &= \text{diag}(Q_1, \ldots, Q_\mathcal{K}), & \mathcal{R} &= \text{diag}(R_1, \ldots, R_\mathcal{K}),
\end{aligned}
\end{equation*}
where $\mathcal{A}, \mathcal{Q} \in \mathbb{R}^{N\times N}$, $\mathcal{B} \in \mathbb{R}^{N\times M}$, and $\mathcal{R}\in \mathbb{R}^{M\times M}$  with 
$$M = \sum_{k=1}^\mathcal{K}m_k \quad \text{and}\quad  N = \sum_{k=1}^{\cal K}n_k.$$
\begin{remark}
Clearly $\mathcal{R}\succ 0$. 
   In addition, under Assumption~\ref{assum:A_B_stabi_A_Q_Obs}, the pair ($\mathcal{A}-\frac{1}{2}\rho I_{N}, \mathcal{B}$) is stabilizable and the pair ($\mathcal{A}-\frac{1}{2}\rho I_{N}, \mathcal{Q}$) is observable  as an immediate consequence.
\end{remark}

\begin{assumption} \label{assum:Hamiltonian_Q}
   The eigenvalues of the Hamiltonian matrix $H_{\Omega}\in \mathbb{R}^{2N \times 2N}$  
   \sg{defined by} 
   \begin{equation}\label{eq:matrixHomega}
      H_{\Omega} = \begin{bmatrix}
         \mathcal{A}- \frac{1}{2}\rho I_N & -\mathcal{BR}^{-1}\mathcal{B}^{\top} \\
          -\mathcal{Q}(I_N-\mathcal{H}) & \quad  -\mathcal{A}^{\top} + \frac{1}{2}\rho I_N
      \end{bmatrix} \in \mathbb{R}^{2N \times 2N}
   \end{equation}
   are strong ($N, N$) c-splitting, and the associated $N$-dimensional stable invariant subspace is a graph subspace.
\end{assumption}
{Following the fixed point approach in \citep{Huang:2007,Huang:2019}, we let the population size in all classes go to infinity such that, with a slight abuse of notation, $\bar{x}_m = \lim_{c_m \to \infty} \frac{1}{c_m}\sum_{i=1}^{c_m} x_{m,i}$, and solve the corresponding limit problem, which generates the following MFG strategy. } 
\begin{proposition} 
\label{prop:key_equations}
Under \sg{Assumptions} \ref{assum:A_B_stabi_A_Q_Obs} and \ref{assum:Hamiltonian_Q}, the 
MFG strategy   
of a generic agent $a_{k, i}$ in class $k$ exists and is uniquely given by 
\begin{equation}\label{eq:optimal_control_sol}
    u^{*}_{k,i}(t) = -R_k^{-1}B_k^\top(P_kx_{k,i} + s_k)
\end{equation}
where $P_k \succ 0$ and $s_k = \sum_{m=1}^{\mathcal{K}}\Pi_{km}\Bar{x}_m$ follow from the solutions to the following algebraic equations 
\begin{subequations}
\label{eq:system_dynamics}
\begin{align}
    \rho\mathcal{P} &= \mathcal{Q} + \mathcal{PA} + \mathcal{A}^\top\mathcal{P} - \mathcal{PBR}^{-1}\mathcal{B}^\top\mathcal{P} \label{eq:dyn_p} \\
    \begin{split}
        \rho\Pi &= - \mathcal{QH}+\Pi(\mathcal{A}-\mathcal{BR}^{-1}\mathcal{B}^\top(\mathcal{P}+\Pi)) \\
        &\quad + (\mathcal{A}^\top - \mathcal{PBR}^{-1}\mathcal{B}^\top)\Pi 
    \end{split} \label{eq:dyn_pi}
\end{align}
\end{subequations} 
and the mean field dynamics
\begin{equation}\label{eq:meanFieldDynamics}
    \dot{\bar{x}} = (\mathcal{A}-\mathcal{BR}^{-1}\mathcal{B}^\top(\mathcal{P}+\Pi))\bar{x}
\end{equation}
with  $\mathcal{P}\in \mathbb{R}^{N\times N}$, $\bar{x} \in \mathbb{R}^{N}$ given by
\begin{equation*}
\begin{aligned}
    \mathcal{P} &= \text{diag}(P_1, \ldots, P_\mathcal{K}),&\bar{x} &= [\bar{x}_1^\top, \ldots, \bar{x}_\mathcal{K}^\top]^\top,
\end{aligned}
\end{equation*}
 and $\Pi = \llbracket \Pi_{km}\rrbracket \in \mathbb{R}^{N\times N}$. 
 
See Appendix A for the detailed proof.
\end{proposition}

By defining $\Omega = \mathcal{P}+\Pi$ and summing the two equations in \eqref{eq:system_dynamics} as in \citep{gao2024linearquadraticmeanfield}, the coupled equations are reduced to a set of decoupled Algebraic Riccati Equations (AREs) in $\mathcal{P}$ and $\Omega$:
\begin{subequations}
\label{eq:system_dynamics_contracted}
\begin{align}
       \rho\mathcal{P} &= \mathcal{Q} + \mathcal{PA} + \mathcal{A}^\top\mathcal{P} - \mathcal{PBR}^{-1}\mathcal{B}^\top\mathcal{P} \label{eq:ARE_P}\\
    \rho\Omega &= \mathcal{Q}(I_N-\mathcal{H}) + \Omega\mathcal{A} + \mathcal{A}^\top\Omega - \Omega\mathcal{BR}^{-1}\mathcal{B}^\top\Omega \label{eq:ARE_Omega}
\end{align}
\end{subequations}
where $\mathcal{P}$ and $\Omega$ can be solved simultaneously and $\Pi$ obtained by $\Pi = \Omega - \mathcal{P}$. 

 {A solution to \eqref{eq:ARE_Omega} is called \emph{stabilizing} if the closed-loop matrix $\mathcal{A}-\frac{1}{2}\rho I_{N}-\mathcal{BR}^{-1}\mathcal{B}^\top \Omega$ is Hurwitz.}
 Assumption~\ref{assum:Hamiltonian_Q}  ensures the existence of a unique stabilizing solution $\Omega$ to \eqref{eq:ARE_Omega}  following   {\citep[Thm.~18]{Huang:2019}}. 

Since all matrices involved in \eqref{eq:ARE_P} are block diagonal, $\mathcal{P}$ can be solved class by class by simultaneously solving for $P_k$ for each class $k$, where $P_k$ is the solution to the following class-specific ARE
\begin{equation} \label{eq:class_spec_ARE}
    \rho P_k = Q_k + P_kA_k + A_k^\top P_k - P_kB_k R_k^{-1}B_k^\top P_k \sg{.}
\end{equation}
\section{Data-driven multi-class LQG-MFG with completely unknown dynamics}
To establish the  multi-class LQG-MFG strategy purely based on trajectory data, we adapt the data-driven approach developed by \cite{Jiang:2012} for linear quadratic control problems and later applied to single class LQG-MFGs by \cite{Xu:23, Xu:25}. This approach combines the Kleinman algorithm \citep{Kleinman:1968} for iteratively solving symmetric ARE and the IRL technique for updating strategies based on trajectory data. 

 \sg{Since the data-driven method  based on \citep{Jiang:2012,Xu:23} requires the ARE to be symmetric, we introduce the following assumption.}
\begin{assumption} \label{assum:QI_H_sym}
    The matrix $\mathcal{Q}(I_N-\mathcal{H})$ is symmetric.
\end{assumption}

 \textbf{Kleinman's Iteration Procedure:}    Under Assumptions~\ref{assum:A_B_stabi_A_Q_Obs}, \ref{assum:Hamiltonian_Q}  and \ref{assum:QI_H_sym}, initially stabilizing gains $L_{P,k}^{(0)}$ and $\mathcal{L}_\Omega^{(0)}$ can be selected. An iterative procedure is then formed to solve equation \eqref{eq:class_spec_ARE} for $P_k$, for all $k$, and \eqref{eq:ARE_Omega} for $\Omega$. 
    Equations \eqref{eq:class_spec_ARE} and \eqref{eq:ARE_Omega} can thus be solved iteratively by
\begin{subequations} 
\label{eq:Kleinman}
\begin{align}
    \rho P_{k}^{(\ell)} &= 
        P_{k}^{(\ell)}(A_k-B_k L_{P,k}^{(\ell-1)}) + (A_k - B_kL_{P,k}^{(\ell-1)})^\top P_{k}^{(\ell)} \notag\\
        &\quad+ (L_{P,k}^{(\ell-1)})^\top R_kL_{P,k}^{(\ell-1)} + Q_k \label{eq:P_update}\\
     \rho\Omega^{(\ell)} &=
        \Omega^{(\ell)}(\mathcal{A}-\mathcal{BL}_\Omega^{(\ell-1)}) + (\mathcal{A} - \mathcal{BL}_\Omega^{(\ell-1)})^\top\Omega^{(\ell)} \notag\\
         &\quad+ (\mathcal{L}_\Omega^{(\ell-1)})^\top\mathcal{R}\mathcal{L}_\Omega^{(\ell-1)} + \mathcal{Q}(I_N-\mathcal{H}) \label{eq:Omega_update}
\end{align}
\end{subequations}
where $L_{P,k}^{(\ell)} = R_k^{-1}B_k^\top P_k^{(\ell)}$ and $\mathcal{L}_\Omega^{(\ell)} = \mathcal{R}^{-1}\mathcal{B}^\top \Omega^{(\ell)}$. $\mathcal{P}$ can then be reconstructed by $\mathcal{P} = \text{diag}(P_1, \ldots, P_\mathcal{K})$.
\begin{proposition}
Under Assumptions \ref{assum:A_B_stabi_A_Q_Obs}, \ref{assum:Hamiltonian_Q}  and \ref{assum:QI_H_sym},  the resulting matrices $P_k^{(\ell)}$, $\Omega_k^{(\ell)}$, $L_{P,k}^{(\ell)}$ and $\mathcal{L}_\Omega^{(\ell)}$  from the Kleinman's iteration procedure above  satisfy the following properties: 
    \begin{enumerate}
    \item $A_k-\frac{1}{2}\rho I_{n_k}-B_kL_{P,k}^{(\ell)}$ and $\mathcal{A}-\frac{1}{2}\rho I_{N} - \mathcal{BL}_{\Omega}^{(\ell)}$ are Hurwitz,
    \item $P_k^* \preceq P_k^{(\ell+1)} \preceq P_k^{(\ell)}$ and $\Omega^* \preceq \Omega^{(\ell+1)} \preceq \Omega^{(\ell)}$,
    \item $\lim_{\ell\rightarrow\infty}L_{P,k}^{(\ell)} = L_{P,k}^*$, $\lim_{\ell\rightarrow\infty}P_k^{(\ell)}=P_k^*$, \\ $\lim_{\ell\rightarrow\infty}\mathcal{L}_{\Omega}^{(\ell)} = \mathcal{L}_{\Omega}^*$, and $\lim_{\ell\rightarrow\infty}\Omega^{(\ell)}=\Omega^*$,\hfill $\square$
\end{enumerate}
where $P_k^*$ is the unique positive definite solution to \eqref{eq:class_spec_ARE} for class $k$ and ${\Omega}^*$ the unique stablizing solution to \eqref{eq:ARE_Omega}.
\end{proposition}

\begin{pf} 
Since the class-specific AREs (\ref{eq:class_spec_ARE}) form $\mathcal{K}$ independent standard Riccati equations associated with LQR problems, the convergence of the iteration for $P_k$ follows directly from \citep{Kleinman:1968}. The ARE for $\Omega$ in (\ref{eq:ARE_Omega}) corresponds to a global LQR problem with a state penalty matrix $\mathcal{Q}(I_{N}-\mathcal{H})$, which may be indefinite. Under Assumptions \ref{assum:A_B_stabi_A_Q_Obs}, \ref{assum:Hamiltonian_Q}, and \ref{assum:QI_H_sym}, the convergence of the policy iteration for $\Omega$ follows from \citep[Lemma 4.1]{Xu:25}, which extends Kleinman's iteration to symmetric indefinite AREs.
\end{pf}

In order to solve the AREs without \sg{the} knowledge of system dynamics, a representative agent $a_{k,1}$ is selected for each class $k$, and $x = [x_{1,1}^\top, x_{2,1}^\top, \ldots, x_{\mathcal{K},1}^\top]^\top \in \mathbb{R}^{N}$ is defined as the augmented states vector, and $u = [u_{1,1}^\top, u_{2,1}^\top, \ldots, u_{\mathcal{K},1}^\top]^\top \in \mathbb{R}^{M}$, the augmented control inputs vector. The value function ansatzs for each $P_k$ and another for $\Omega$ are respectively defined as follows:
\begin{subequations}
\label{eq:val_ansatz}
\begin{align}
    \Psi_{1,k}(t,x_{k,1}) &= e^{-\rho t}x_{k,1}^\top P_k^{(\ell)} x_{k,1} \label{eq:val_local} \\
    \Psi_2(t, x) &= e^{-\rho t}x^\top \Omega^{(\ell)}x. \label{eq:val_global}
\end{align}
\end{subequations}
Under Assumption \ref{assum:A_B_stabi_A_Q_Obs}, a control input 
\begin{align}
    u(t) = \alpha(t) &= [\alpha_{1,1}^\top(t),\ldots,\alpha_{\mathcal{K},1}^\top(t)]^\top \notag \\&= - \mathcal{L}_{\Omega}^{(0)}x(t) + l(t)
\end{align}
can be selected, where $\alpha(t)$ denotes  the control input during the learning phase, composed of an initially stabilizing global gain $\mathcal{L}_{\Omega}^{(0)}$ such that $\mathcal{A} -\frac{1}{2}\rho I_{N} - \mathcal{B}\mathcal{L}_{\Omega}^{(0)}$ is Hurwitz, and a global exploratory noise $l(t) = [l_{1,1,1}(t),\ldots,l_{\mathcal{K},1, m_k}(t)]^\top \in \mathbb{R}^{M}$, composed of independent noise inputs for all elements of $u(t)$. The associated $L_{k}^{(0)}$ included in $\mathcal{L}_\Omega^{(0)}$ are also selected to be initially stabilizing gain for each class $k$ such that $A_k -\frac{1}{2}\rho I_{n_k} - B_kL_{k}^{(0)}$ are Hurwitz. The dynamics of the representative agents and the augmented states are given by
\begin{subequations}
\label{eq:dyn_sys_explo}
    \begin{align}
    dx_{k,1}(t) &= (A_{k}x_{k,1}(t) + B_k\alpha_{k,1}(t))dt + D_kdw_{k,1}(t) \\
    dx(t) &= \left(\mathcal{A}x(t) + \mathcal{B}\alpha(t)\right)dt + \mathcal{D}dw
\end{align}
\end{subequations}
where 
\begin{align*}
\mathcal{D} & = \text{diag}(D_1, \ldots, D_\mathcal{K})\text{ and }  
w(t)  = [w_{1,1}^\top(t), \ldots, w_{\mathcal{K},1}^\top(t)]^\top.
\end{align*}
Applying Itô's formula to (\ref{eq:val_ansatz}), and using  \eqref{eq:Kleinman} and (\ref{eq:dyn_sys_explo}), we obtain
\allowdisplaybreaks
\begin{subequations}
\label{eq:InfinitesimalGenerators}
\begin{align}
    d\Psi_{1,k} &= 
         e^{-\rho t} \Big[ -x_{k,1}^\top (L_{P,k}^{(\ell-1)})^\top R_k L_{P,k}^{(\ell-1)} x_{k,1} \label{eq:dPsi_local}\\
        &\quad - x_{k,1}^\top Q_k x_{k,1} + 2(\alpha_{k,1} + L_{P,k}^{(\ell-1)} x_{k,1})^\top R_k L_{P,k}^{(\ell)} x_{k,1} \notag\\
        &\quad + \operatorname{Tr}(D_k D_k^\top P_k^{(\ell)}) \Big] dt + 2e^{-\rho t} x_{k,1}^\top P_k^{(\ell)} D_k dw_{k,1}  \notag\\[5pt] 
    d\Psi_{2} &= 
        e^{-\rho t} \Big[ -x^\top (\mathcal{L}_\Omega^{(\ell-1)})^\top \mathcal{R} \mathcal{L}_\Omega^{(\ell-1)} x \label{eq:dPsi_global}\\
        &\quad - x^\top \mathcal{Q}(I_N - \mathcal{H}) x + 2(\alpha + \mathcal{L}_\Omega^{(\ell-1)} x)^\top \mathcal{R} \mathcal{L}_\Omega^{(\ell)} x \notag\\
        &\quad + \operatorname{Tr}(\mathcal{D} \mathcal{D}^\top \Omega^{(\ell)}) \Big] dt + 2e^{-\rho t} x^\top \Omega^{(\ell)} \mathcal{D} dw.\notag
\end{align}
\end{subequations}
\sg{Let $\Delta t$ denote the length of integration chosen for transforming the trajectory data.} Integrating both sides for a time interval $[t, t+\Delta t]$ and taking its expectation yield
\begin{subequations}
\label{eq:expectation_integrals}
\begin{align}
    &\mathbb{E}\big[e^{-\rho (t+\Delta t)}x_{k,1}^\top(t+\Delta t)P_k^{(\ell)} x_{k,1}(t+\Delta t) \nonumber \\
    &\qquad\qquad\qquad\qquad\qquad\qquad\quad - e^{-\rho t}x_{k,1}^\top(t)P_k^{(\ell)} x_{k,1}(t)\big] \nonumber \\
    &= -\mathbb{E}\int_t^{t+\Delta t}e^{-\rho \tau}x_{k,1}^\top\big((L_{P,k}^{(\ell-1)})^\top R_k L_{P,k}^{(\ell-1)} + Q_k\big)x_{k,1}d\tau \nonumber \\
    &\quad + 2\mathbb{E}\int_t^{t+\Delta t}e^{-\rho \tau}(\alpha_{k,1} + L_{P,k}^{(\ell-1)}x_{k,1})^\top R_k L_{P,k}^{(\ell)}x_{k,1} d\tau \nonumber \\
    &\quad + \tfrac{1}{\rho}(e^{-\rho t} - e^{-\rho (t+\Delta t)}) \operatorname{Tr}(D_k D_k^\top P_k^{(\ell)}), \\[5pt]
    &\mathbb{E}\big[e^{-\rho (t+\Delta t)}x^\top(t+\Delta t)\Omega^{(\ell)}x(t+\Delta t) - e^{-\rho t}x^\top(t)\Omega^{(\ell)}x(t)\big] \nonumber \\
    &= -\mathbb{E}\int_t^{t+\Delta t}e^{-\rho \tau}x^\top\big((\mathcal{L}_{\Omega}^{(\ell-1)})^\top\mathcal{R}\mathcal{L}_{\Omega}^{(\ell-1)} + \mathcal{Q}(I_{N}-\mathcal{H})\big)xd\tau \nonumber \\
    &\quad + 2\mathbb{E}\int_t^{t+\Delta t}e^{-\rho \tau}(\alpha + \mathcal{L}_\Omega^{(\ell-1)}x)^\top \mathcal{R}\mathcal{L}_\Omega^{(\ell)}x d\tau \nonumber \\
    &\quad + \tfrac{1}{\rho}(e^{-\rho t} - e^{-\rho (t+\Delta t)}) \operatorname{Tr}(\mathcal{D} \mathcal{D}^\top \Omega^{(\ell)}).
\end{align}
\end{subequations}
Equations (\ref{eq:expectation_integrals}) can be expressed using the Kronecker product representation, which yields
\allowdisplaybreaks
\begin{subequations}
\label{eq:IRL_Identities_kron}
\begin{align}
    0 &= {(\delta_{P,k}^t})^\top \hat{P}_{k}^{(\ell)} + \delta_\rho^t \theta_{P,k}^{(\ell)}  \label{eq:Kron_Pk}\\
    &\quad + ({\mathcal{I}_{xx,k}^t})^\top \operatorname{vec}((L_{P,k}^{(\ell-1)})^\top R_k L_{P,k}^{(\ell-1)} + Q_k) \notag \\
    &\quad - 2\Big[({\mathcal{I}_{x\alpha,k}^t})^\top(I_{n_k} \otimes R_k) \notag \\
    &\quad\quad + {(\mathcal{I}_{xx,k}^t})^\top(I_{n_k} \otimes (L_{P,k}^{(\ell-1)})^\top R_k)\Big]\operatorname{vec}(L_{P,k}^{(\ell)}) \notag\\
    0 &= ({\delta_{\Omega,k}^t})^\top \hat{\Omega}^{(\ell)} + \delta_\rho^t \theta_{\Omega}^{(\ell)} \label{eq:Kron_Omega}\\
    &\quad + ({\mathcal{I}_{xx}^t})^\top \operatorname{vec}((\mathcal{L}_{\Omega}^{(\ell-1)})^\top \mathcal{R}\mathcal{L}_{\Omega}^{(\ell-1)} + \mathcal{Q}(I_{N}-\mathcal{H})) \notag \\
    &\quad - 2\Big[({\mathcal{I}_{x\alpha}^t})^\top(I_{N} \otimes \mathcal{R}) \notag\\
    &\quad\quad + {(\mathcal{I}_{xx}^t})^\top(I_{N} \otimes (\mathcal{L}_{\Omega}^{(\ell-1)})^\top \mathcal{R})\Big]\operatorname{vec}(\mathcal{L}_\Omega^{(\ell)}) \notag
\end{align}
\end{subequations}
where the terms for the equation for \eqref{eq:Kron_Pk} are 
\begin{flalign*}
    &\begin{aligned}
        \hat{P}_k^{(\ell)} &= \big[P_{k,11}^{(\ell)}, 2P_{k,12}^{(\ell)}, \ldots, 2P^{(\ell)}_{k,1n_k}, \\
        &\qquad P_{k,22}^{(\ell)}, \ldots, P_{k,n_kn_k}^{(\ell)}\big]^\top \in \mathbb{R}^{\frac{n_k(n_k+1)}{2}} 
    \end{aligned} &\\
    &\begin{aligned}
        \hat{x}_{k,1} &= \big[x_1^2, x_1x_2, \ldots, x_1x_{n_k}, x_2^2, \ldots, x_{n_k}^2\big]^\top \in \mathbb{R}^{\frac{n_k(n_k+1)}{2}},
    \end{aligned}
    \end{flalign*}
    where $x_j = x_{k,1,j}$
    \begin{flalign*}
    &\delta^t_{P,k} = \mathbb{E}\big[e^{-\rho (t+\Delta t)}\hat{x}_{k,1}(t+\Delta t) - e^{-\rho t}\hat{x}_{k,1}(t)\big] \in \mathbb{R}^{\frac{n_k(n_k+1)}{2}}, &\\
    &\mathcal{I}_{xx,k}^t = \mathbb{E}\left[\int_t^{t+\Delta t}e^{-\rho \tau} (x_{k,1} \otimes x_{k,1}) d\tau\right] \in \mathbb{R}^{n_k^2}, &\\
    &\mathcal{I}_{x\alpha, k}^t = \mathbb{E}\left[\int_t^{t+\Delta t}e^{-\rho \tau} (x_{k,1} \otimes \alpha_{k,1}) d\tau\right] \in \mathbb{R}^{n_km_k}, &\\
    &\delta^t_\rho = e^{-\rho (t+\Delta t)} - e^{-\rho t} \in \mathbb{R},\: \theta_{P,k}^{(\ell)} = \frac{1}{\rho}\operatorname{Tr}(D_k D_k^\top P_k^{(\ell)}) \in \mathbb{R},
\end{flalign*}
and the terms for \eqref{eq:Kron_Omega} are
\begin{align*}
    \hat{\Omega}^{(\ell)} &= \big[\Omega_{11}^{(\ell)}, 2\Omega_{12}^{(\ell)}, \ldots, 2\Omega_{1N}^{(\ell)}, \Omega_{22}^{(\ell)}, \ldots, \Omega_{NN}^{(\ell)}\big]^\top \in \mathbb{R}^{\frac{N(N+ 1)}{2}}, \\
    \hat{x} &= \big[x_1^2, x_1x_2, \ldots, x_1x_{N}, x_2^2, \ldots, x_{N}^2\big]^\top \in \mathbb{R}^{\frac{N(N+1)}{2}}, \\
    \delta^t_{\Omega} &= \mathbb{E}\big[e^{-\rho (t+\Delta t)}\hat{x}(t+\Delta t) - e^{-\rho t}\hat{x}(t)\big] \in \mathbb{R}^{\frac{N(N+1)}{2}}, \\
    \mathcal{I}_{xx}^t &= \mathbb{E}\left[\int_t^{t+\Delta t}e^{-\rho \tau} (x \otimes x) d\tau\right] \in \mathbb{R}^{N^2}, \\
    \mathcal{I}_{x\alpha}^t &= \mathbb{E}\left[\int_t^{t+\Delta t}e^{-\rho \tau} (x \otimes \alpha) d\tau\right] \in \mathbb{R}^{NM}, \\
    \omit\rlap{$\delta_\rho^t = e^{-\rho (t+\Delta t)} - e^{-\rho t} \in \mathbb{R} \text{ and } \theta_{\Omega, k}^{(\ell)} = \tfrac{1}{\rho}\operatorname{Tr}(\mathcal{D}\mathcal{D}^\top\Omega^{(\ell)}) \in \mathbb{R}.$}
\end{align*}
Using $l \in \mathbb{N}$ time steps associated with real-time data, the following matrices are defined
\begin{align*}
    \Delta_{1k} &= [\delta_{P,k}^{t_1}, \ldots, \delta_{P,k}^{t_l}]^\top \in \mathbb{R}^{l \times \frac{n_k(n_k+1)}{2}}, \\
    \Delta_1 &= [\delta_{\Omega}^{t_1}, \ldots, \delta_{\Omega}^{t_l}]^\top \in \mathbb{R}^{l \times \frac{N(N+1)}{2}}, \\
    \Delta_{2k} &= -2[\mathcal{I}_{xx, k}^{t_1}, \ldots, \mathcal{I}_{xx,k}^{t_l}]^\top(I_{n_k} \otimes (L_{P,k}^{(k-1)})^\top R_k) \\
    &\quad -2[\mathcal{I}_{x\alpha, k}^{t_1}, \ldots, \mathcal{I}_{x\alpha, k}^{t_l}]^\top(I_{n_k} \otimes R_k) \in \mathbb{R}^{l \times m_kn_k}, \\
    \Delta_3 &= -2[\mathcal{I}_{xx}^{t_1}, \ldots, \mathcal{I}_{xx}^{t_l}]^\top(I_{N} \otimes (\mathcal{L}_{\Omega}^{(\ell-1)})^\top \mathcal{R}) \\
    &\quad -2[\mathcal{I}_{x\alpha}^{t_1}, \ldots, \mathcal{I}_{x\alpha}^{t_l}]^\top(I_{N} \otimes \mathcal{R}) \in \mathbb{R}^{l \times NM}, \\
    \Delta_{4k} &= -[\mathcal{I}_{xx,k}^{t_1}, \ldots, \mathcal{I}_{xx,k}^{t_l}]^\top \\
    &\quad \cdot \operatorname{vec}((L_{P,k}^{(\ell-1)})^\top R_k L_{P,k}^{(\ell-1)} + Q_k) \in \mathbb{R}^l, \\
    \Delta_5 &= -[\mathcal{I}_{xx}^{t_1}, \ldots, \mathcal{I}_{xx}^{t_l}]^\top \\
    &\quad \cdot \operatorname{vec}((\mathcal{L}_{\Omega}^{(\ell-1)})^\top \mathcal{R} \mathcal{L}_{\Omega}^{(\ell-1)} + \mathcal{Q}(I_{N}-\mathcal{H})) \in \mathbb{R}^l, \\
    \Delta_6 &= [\delta_\rho^{t_1}, \ldots, \delta_\rho^{t_l}]^\top \in \mathbb{R}^l.
\end{align*}
and equations (\ref{eq:IRL_Identities_kron}), using the data from $l$ time steps, can be expressed as
\begin{subequations} \label{eq:time_step_data}
\begin{align}
0 &= \Delta_{1k}\hat{P}_k^{(\ell)} + \Delta_{2k}\text{vec}(L_{P,k}^{(\ell)}) + \Delta_6\theta^{(\ell)}_{P,k} - \Delta_{4k} \\
0 &= \Delta_{1}\hat{\Omega}_k^{(\ell)} + \Delta_3\text{vec}(\mathcal{L}_{\Omega}^{(\ell)}) + \Delta_6\theta_{\Omega}^{(\ell)} - \Delta_5.
\end{align}
\end{subequations}
{Equation} \eqref{eq:time_step_data} can then be reformulated as the matrix form
\begin{subequations}
\label{eq:LeastSquaresSystem}
\begin{align}
    \underbrace{\begin{bmatrix} \Delta_{1k} & \Delta_{2k} & \Delta_6 \end{bmatrix}}_{\Xi_{1,k}} 
    \begin{bmatrix} \hat{P}^{(\ell)}_k \\ \operatorname{vec}(L_{P,k}^{(\ell)}) \\ \theta_{P,k}^{(\ell)} \end{bmatrix} &= \Delta_{4k} \label{eq:LS_local} \\[6pt]
    \underbrace{\begin{bmatrix} \Delta_{1} & \Delta_{3} & \Delta_6 \end{bmatrix}}_{\Xi_{2}} 
    \begin{bmatrix} \hat{\Omega}^{(\ell)} \\ \operatorname{vec}(\mathcal{L}_{\Omega}^{(\ell)}) \\ \theta_{\Omega}^{(\ell)} \end{bmatrix} &= \Delta_{5}. \label{eq:LS_global}
\end{align}
\end{subequations}
We introduce the following assumption regarding the requirement for the trajectory data.
\begin{assumption}\label{assum:persExcit}
There exists an integer $L > 0$ such that for $l \geq L$, the matrices 
\begin{align*}
    \begin{bmatrix}
        \mathcal{I}_{xx,k}^{t_1} & \mathcal{I}_{xx,k}^{t_2} & \ldots & \mathcal{I}_{xx,k}^{t_l}\\
        \mathcal{I}_{x\alpha,k}^{t_1} & \mathcal{I}_{x\alpha,k}^{t_2} & \ldots & \mathcal{I}_{x\alpha,k}^{t_l}\\
        \delta_\rho^{t_1} & \delta_\rho^{t_2} & \ldots & \delta_\rho^{t_l}
    \end{bmatrix}
\end{align*}
are of rank $\frac{n_k(n_k+1)}{2} + m_kn_k + 1$ for all class $k$ and
\begin{align*}
    \begin{bmatrix}
        \mathcal{I}_{xx}^{t_1} & \mathcal{I}_{xx}^{t_2} & \ldots & \mathcal{I}_{xx}^{t_l}\\
        \mathcal{I}_{x\alpha}^{t_1} & \mathcal{I}_{x\alpha}^{t_2} & \ldots & \mathcal{I}_{x\alpha}^{t_l}\\
        \delta_\rho^{t_1} & \delta_\rho^{t_2} & \ldots & \delta_\rho^{t_l}
    \end{bmatrix}
\end{align*}
is of rank $\frac{N(N+1)}{2} + NM + 1$.
\end{assumption}
\begin{remark}
Assumption \ref{assum:persExcit} ensures that \eqref{eq:LeastSquaresSystem} have unique solutions. In practice, it can be satisfied by injecting exploration noise into the control input during the learning phase, such as Gaussian noise, or a sum of sinusoids, as shown in \citep{Jiang:2012} and \citep{Xu:23}.
\end{remark}
The multi-class LQG-MFG integral reinforcement learning can thus be formulated in Algorithm~\ref{alg:MFG_IRL}.
\begin{algorithm}[h]
\caption{Multi-class Mean Field Games IRL}
\label{alg:MFG_IRL}
\begin{algorithmic}
\Statex \textbf{Initialization:} 
\State Choose initially stabilizing gains $\mathcal{L}_{\Omega}^{(0)}$ and $L_{P,k}^{(0)}$.
\State Select representative agents $a_{k,1}$ for all class $k$, set threshold $\varepsilon$, and set iteration counter $\ell=1$.
\Statex \vspace{5pt} \textbf{Data Collection:}
\State Apply $u(t) = - \mathcal{L}_{\Omega}^{(0)}x(t)+ l(t)$ to a global system of representative agents.
\State Compute data vectors until $\text{rank}(\Xi_{1,k})$ and $\text{rank}(\Xi_{2})$ satisfy persistence of excitation. (Assumption \ref{assum:persExcit})
\Statex \vspace{5pt} \textbf{Policy Iteration:}
\Repeat
    \State Solve for class-specific parameters:
    \begin{equation} \label{eq:leastsquaresClass}
        \begin{bmatrix} \hat{P}_k^{(\ell)} \\ \text{vec}(L_{P,k}^{(\ell)}) \\ \theta_{P,k}^{(\ell)} \end{bmatrix} = (\Xi_{1,k}^\top \Xi_{1,k})^{-1}\Xi_{1,k}^\top \Delta_{4k}, \quad \text{for all } k.
    \end{equation}
    \State Solve for global parameters:
    \begin{equation} \label{eq:leastsquaresGlobal}
        \begin{bmatrix} \hat{\Omega}^{(\ell)} \\ \text{vec}(\mathcal{L}_{\Omega}^{(\ell)}) \\ \theta_{\Omega}^{(\ell)} \end{bmatrix} = (\Xi_{2}^\top \Xi_2)^{-1}\Xi_2^\top \Delta_5.
    \end{equation}
    \State Update $\ell \leftarrow \ell+1$
\Until{$\|P_{k}^{(\ell)} - P_k^{(\ell-1)}\| \leq \varepsilon$ \text{for all $k$ }\textbf{\&} $\|\Omega^{(\ell)} - \Omega^{(\ell-1)}\| \leq \varepsilon$}
\end{algorithmic}
\end{algorithm}
\begin{remark}[Trajectory Data Required by the Algorithm]
   The data-driven algorithm  requires multiple trajectories of the states and control inputs of the representative agent $a_{k,1}$ of class $k$ for all $k\in\{1,\cdots , \mathcal{K}\}$. 
   These trajectories are defined for the interval $[t_1, t_l+\Delta t]$ and satisfy Assumption~\ref{assum:persExcit}.
   Using a large finite number of trajectories, the mean of data integrals in $\delta_{P,k}$,  $\mathcal{I}_{xx,k}$,  $\mathcal{I}_{x\alpha,k}$, $\delta_\Omega$, $\mathcal{I}_{xx}$ and $\mathcal{I}_{x\alpha}$ can approximate their expectations.
\end{remark}

\section{Numerical Example}
A numerical simulation for learning multi-class LQG-MFG gain matrices using   Algorithm \ref{alg:MFG_IRL} is carried out. The parameters of the problems  are given in  Table~\ref{tab:system_params}.
%
\begin{table}[h]
\centering
\normalsize 
\caption{System Parameters}
\label{tab:system_params}
\setlength{\tabcolsep}{8pt} 
\renewcommand{\arraystretch}{1.5} 
\begin{tabular}{@{} l c c c @{}}
\toprule
\textbf{Param} & \textbf{Class 1} & \textbf{Class 2} & \textbf{Class 3} \\ \midrule
$A_k$ & $\left[ \begin{smallmatrix} 0 & 10 \\ -10 & -3 \end{smallmatrix} \right]$ & $\left[ \begin{smallmatrix} 0 & 1 & 0 \\ 0 & 0 & 1 \\ -2 & -3 & -5 \end{smallmatrix} \right]$ & $\left[ \begin{smallmatrix} 0 & -4 \\ 3 & -6 \end{smallmatrix} \right]$ \\
$B_k$ & $\left[ \begin{smallmatrix} 1 \\ 1 \end{smallmatrix} \right]$ & $\left[ \begin{smallmatrix} 0 & 0 \\ 0 & 1 \\ 1 & 0.5 \end{smallmatrix} \right]$ & $\left[ \begin{smallmatrix} 0.8 \\ 3 \end{smallmatrix} \right]$ \\
$D_k$ & $0.1 I_{2}$ & $0.1 I_{3}$ & $0.1 I_{2}$ \\
$Q_k$ & $\text{diag}(20, 10)$ & $\text{diag}(10, 15, 20)$ & $\text{diag}(30, 20)$ \\
$R_k$ & $0.8$ & $\text{diag}(0.5, 0.7)$ & $0.6$ \\ \bottomrule
\end{tabular}
\end{table}
The global interaction matrix $\mathcal{H}$ is given in Fig.~\ref{fig:sidebyside}, which ensures that Assumptions \ref{assum:Hamiltonian_Q} and \ref{assum:QI_H_sym} are respected.
\begin{figure}[H]
     \centering
     \vspace{-5pt}
     \begin{subfigure}[b]{0.24\textwidth}
         \centering
         \includegraphics[width=\textwidth]{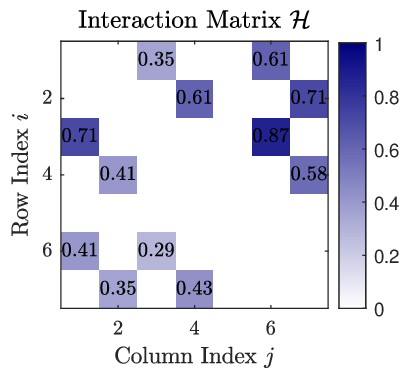}
         \label{fig:second}
     \end{subfigure}
     \hfill
     \begin{subfigure}[b]{0.24\textwidth}
         \centering
         \includegraphics[width=\textwidth]{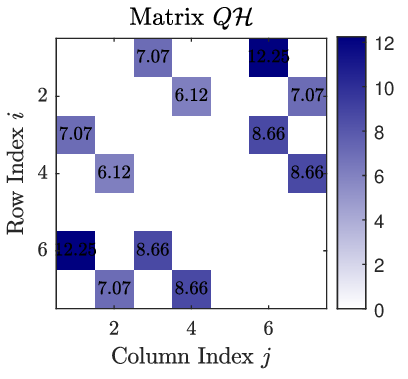}
         \label{fig:first}
     \end{subfigure}
     \vspace{-30pt}
     \caption{Coupling matrix $\mathcal{H}$ and associated symmetric $\mathcal{QH}$}
     \label{fig:sidebyside}
     \vspace{-10pt}
\end{figure}
Under Assumptions \ref{assum:A_B_stabi_A_Q_Obs}-
\ref{assum:QI_H_sym}, the AREs \eqref{eq:class_spec_ARE} and \eqref{eq:ARE_Omega} \sg{admit} unique positive-definite solutions $P_k$ and $\Omega$\sg{, respectively}.

To run Algorithm \ref{alg:MFG_IRL}, a representative agent $a_{k,1}$ is selected for each class $k$. Initial states and initial gains are the same for each class, with
\begin{align*}
    L_k^{(0)} = \mathbf{0}_{m_k\times n_k}, \qquad x_{k,1}(0) = \mathbf{1}_{n_k\times 1},
\end{align*}
and $\mathcal{L}_\Omega^{(0)} = \text{diag}(L_1^{(0)}, \ldots, L_k^{(0)})$.

An exploratory control input $u(t) = \alpha(t) = -\mathcal{L}_{\Omega}^{(0)}x(t)+ l(t)$ is used. For each class $k$ and input channel $i$, the signal is defined as $l_{k,1,i}(t)=\sum_{j=1}^{500} A_e \sin(\omega_j t)$, with $\omega_j \sim \mathcal{U}([-100,100])$ i.i.d. and $A_e=25$. The global system composed of agents $a_{k,1}$ runs 20 times for 20 s. Algorithm~1 is then executed with $\varepsilon=10^{-9}$. Fig.~\ref{fig:convergence_combined} illustrates the convergence of $(L_{P,k},P_k)$ and $(\mathcal{L}_\Omega,\Omega)$ to their true values.

\begin{figure}[h]
    \centering
    \includegraphics[width=0.95\columnwidth]{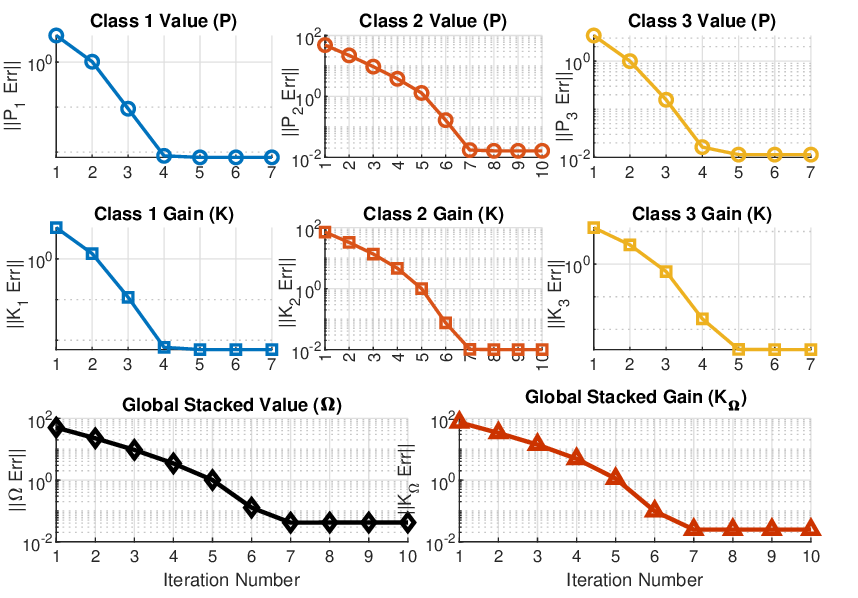}
    \caption{Convergence analysis of the IRL algorithm for local class-level results and global system-level results}
    \label{fig:convergence_combined}
\end{figure}

By iteration 11, all matrices have converged to the fixed threshold $\varepsilon$, and the Frobenius norm error of the learned matrices to their ground truth obtained using MATLAB's \texttt{care()} function is small.

Similar to \cite{Xu:23}, mean-field trajectories are computed offline using the learned gains by simulating $N_s=100$ trajectories of the augmented system with $\alpha=-\mathcal{L}_\Omega x$ and $x(0)=\mathbf{1}_{7\times 1}$. The empirical mean field is $x^{(N_s)}(t)=\frac{1}{N_s}\sum_{j=1}^{N_s}x_{(j)}(t)$, where $x_{(j)}$ is the $j$-th trajectory, which converges to $\bar{x}(t)$ by the law of large numbers. Class-wise mean fields are extracted as subvectors.
\begin{figure}[h]
    \centering
    \includegraphics[width=0.95\linewidth]{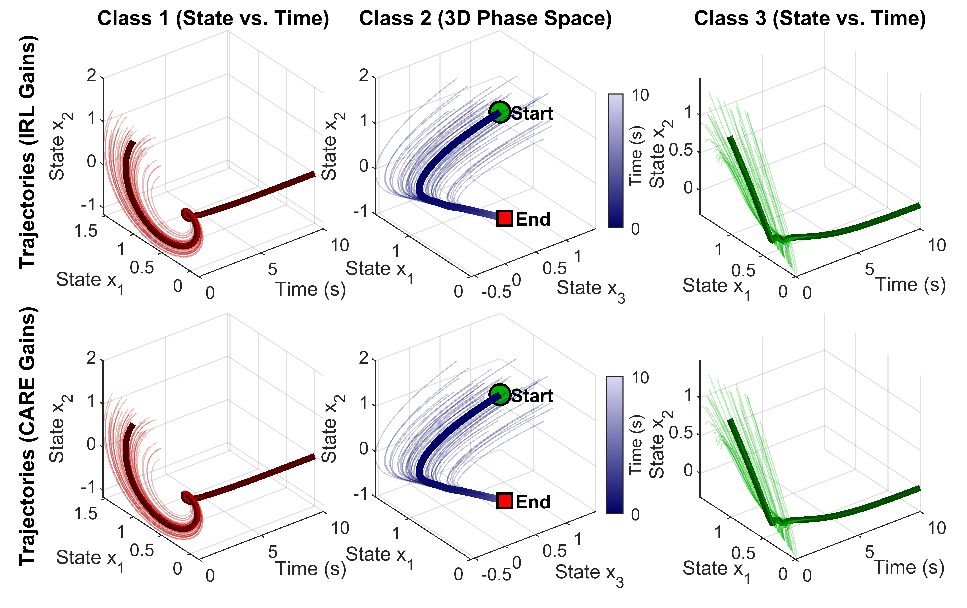}
    \caption{State trajectories (Classes 1 \& 3) and 3D phase space representation (Class 2) simulated under data-driven learned controls (top row) versus MATLAB's \texttt{care()} gains (bottom row).}
    \label{fig:mftrajectories}
\end{figure}
For validation, a finite population of 50 agents per class is simulated, with all initial states located uniformly in [0.5, 1.5]. The agents are simulated using control input
\begin{equation}
        u_{k,i}(t) = -L_{P,k}x_{k,i}(t) - L_{\Pi, k}x^{(50)}(t),
\end{equation}
where $L_{\Pi,k} \in \mathcal{R}^{m_k\times N}$ is the $k$-th block row of the matrix $\mathcal{L}_{\Pi} = \mathcal{L}_{\Omega} -\mathcal{L}_{\mathcal{P}} \in \mathbb{R}^{M\times N}$, with $\mathcal{L}_{\mathcal{P}} = \text{diag}(L_{P,1}, \ldots, L_{P,\mathcal{K}}) \in \mathbb{R}^{M \times N}$. 
Fig. \ref{fig:mftrajectories} compares trajectories under learned gains and MATLAB's \texttt{care()} gains.
\section{Conclusion}
This paper established a data-driven algorithm for computing  strategies for continuous-time infinite-horizon LQG MFGs with heterogeneous network-coupled populations that contain completely unknown dynamics. During the data collection and learning phase, a global system is formed for a generic agent in each class, and an exploration noise is applied to an initially stabilizing control to ensure persistency of excitation. Under conditions  on the persistency of excitation and  on the existence of unique stabilizing solution for the corresponding AREs,  the algorithm converges to the MFG strategies that depend on the classes and  network couplings. Future investigations should extend the solutions 
to the cases with network-coupled dynamics,  directed network couplings, 
finite time horizons, and  nonlinear agent dynamics, and apply the algorithm in applications including renewable energy systems with user populations. 




\bibliography{IEEEabrv, ifacconf}
\appendix
\section{Proof of proposition \ref{prop:key_equations}}
Let the population size in all classes go to infinity. Then, in the limit problem, an individual agent is negligible in the mean field of every class. Thus,  the mean fields of all classes can be treated as deterministic trajectories in the limit problem.
For an agent $\alpha_{k,i}$, consider the augmented state $z_{k,i} = [x_{k,i}^\top, \bar{\mathcal{X}}^\top]^\top$ where $\bar{\mathcal{X}} = [\bar{x}_{1}^\top, \ldots,\bar{x}_{k}^\top]^\top$. Then the dynamics of the augmented states $z_{k,i}$ satisfy 
\begin{align}
dz_{k,i}&=\left(\tilde{A}_k z_{k,i} + \tilde{B}_k u_{k,i}\right)dt + \tilde{D}_k dw_{k,i}
\end{align}
where 
\[
\tilde{A}_k = \begin{bmatrix}
        A_k &\mathbf{0}_{n_k\times N}\\
        \mathbf{0}_{N\times n_k} & \bar{G}
    \end{bmatrix}, \tilde{B}_k = \begin{bmatrix}
        B_k \\ \mathbf{0}_{N\times m_k}
    \end{bmatrix},  \tilde{D}_k =  \begin{bmatrix}
        D_k \\ \mathbf{0}_{N\times d_k}
    \end{bmatrix},
\]
$\bar{G}$ is the drift of the mean field, assumed to be known. The linear mean-field dynamics with the drift $\bar{G}$ will be identified later from the consistency condition required by the  fixed-point approach for MFGs \citep{Huang:2007}.

The problem can then be framed as a standard LQR problem with cost functional 
\begin{equation}
J_k = \mathbb{E} \int_0^\infty e^{-\rho t} \left( \tilde{x}_{k,i}^\top Q_k \tilde{x}_{k,i} + u_{k,i}^\top R_k u_{k,i} \right) dt.
\end{equation}
where $\tilde{x}_{k,i} = x_{k,i} - \sum_{m=1}^{\mathcal{K}}H_{km}\Bar{x}_m$ and  with a slight abuse of notation $\bar{x}_m = \lim_{c_m \to \infty} \frac{1}{c_m}\sum_{i=1}^{c_m} x_{m,i}$. 
Then, assuming a standard value ansatz for an LQR problem, the class-specific value ansatz is 
\begin{align*}
    V_k (z_{k,i}) &= \begin{bmatrix}
        x_{k,i}\\\bar{\mathcal{X}}
    \end{bmatrix}^\top \begin{bmatrix}
        P_{11,k} & P_{12,k} \\ P_{21,k} & P_{22,k}
    \end{bmatrix} \begin{bmatrix}
        x_{k,i} \\\ \bar{\mathcal{X}}
    \end{bmatrix} \triangleq z_{k,i}^\top \tilde{P}_k z_{k,i}
\end{align*}
where $P_{11,k} \in \mathbb{R}^{n_k\times n_k}, P_{12,k} = P_{21,k}^\top \in \mathbb{R}^{n_k\times N},$ and $P_{22,k} \in \mathbb{R}^{N\times N}$. The infinite-time Hamilton-Jacobi-Bellman (HJB) equation is thus
\begin{equation*}
\begin{split}
    \rho V_k(z_{k,i}) 
    =  & \inf_{u_{k,i}} \{L(\tilde{x}_{k,i}, u_{k,i}) + \nabla_{z_{k,i}}V_k(z_{k,i})^\top f(z_{k,i}, u_{k,i})\}\\ &\qquad+ \text{Tr}(\frac{1}{2}\tilde{D}_k\tilde{D}_k^\top V_k(z_{k,i}))
    \end{split}
\end{equation*}
where 
\begin{align*}
   L(\tilde{x}_{k,i}, u_{k,i})  &= \tilde{x}_{k,i}^\top Q_k \tilde{x}_{k,i} + u_{k,i}^\top R_k u_{k,i}\\
    f(z_{k,i}, u_{k,i})  & = \tilde{A}_k z_{k,i} + \tilde{B}_k u_{k,i}.
\end{align*}
Expanding and replacing the terms in the HJB, it becomes 
\begin{equation*}
    \begin{split}
       \rho z_{k,i}^\top \tilde{P}_{k} z_{k,i} 
         =  &  \inf_{u(\cdot)}\{z_{k,i}^\top\tilde{Q}_kz_{k,i} + u_{k,i}R_k u_{k,i} + 2z_{k,i}^\top\tilde{P}_k\tilde{A}_k z_{k,i}\\& \quad+ 2z_{k,i}^\top \tilde{P}_k\tilde{B}_k u_{k,i}\} + \text{Tr}(\frac{1}{2}\tilde{D}_k\tilde{D}_k^\top V_k(z_{k,i}))
    \end{split}
\end{equation*}
where 
\begin{align*}
    \tilde{Q}_k &= \begin{bmatrix}
    Q_k & -Q_kH_k\\-H_k^\top Q_k & H_k^\top Q_k H_k
    \end{bmatrix}\\
    H_k &= [H_{k1},\ldots,H_{k\mathcal{K}}] \in \mathbb{R}^{n_k \times N}.
\end{align*}
Setting the gradient with respect to $u$ to zero yields the optimal control law (i.e. the best response for the limit MFG problem)
\begin{align}
    u_{k,i}^* &= -R_k^{-1}\tilde{B}_k^\top \tilde{P}_k z_{k,i}\\
    &= -R_k^{-1}B_k^\top(P_{11,k}x_{k,i} + P_{12,k}\bar{\mathcal{X}}).
\end{align}
The dynamics of the class-specific mean-field associated with class $k$ are obtained by 
\begin{align}
    \dot{\bar{x}}_{k} &= A_k\bar{x}_k - B_k R_k^{-1}B_k^\top(P_{11,k}\bar{x}_k + P_{12,k}\bar{\mathcal{X}})
\end{align}
and the dynamics of $\bar{\mathcal{X}}$ are given by
\begin{equation}
    \dot{\bar{\mathcal{X}}} = (\mathcal{A} - \mathcal{BR}^{-1}\mathcal{B}^\top(\mathcal{P}_{11}+\mathcal{P}_{12}))\bar{\mathcal{X}}
\end{equation}
where $\mathcal{A, B,R}$ and $\mathcal{P}_{11}$ are block diagonal matrices composed of matrices $A_k, B_k, R_k$ and $P_{11,k}$ for each class $k$, and $\mathcal{P}_{12} = [P_{12,1}^\top,\ldots, P_{12,k}^\top]^\top$.
Thus, the consistency condition for the mean field in the fixed-point approach \citep{Huang:2007} is then equivalently given by $$\bar{G} = \mathcal{A} - \mathcal{BR}^{-1}\mathcal{B}^\top(\mathcal{P}_{11}+\mathcal{P}_{12}).$$
Plugging in the optimal control back into the HJB and matching the coefficients with the value function ansatz yield the infinite-time discounted cost Riccati equation 
\begin{equation}
\label{eq:diffDE}
    \begin{split}
        &\rho \tilde{P}_k = \tilde{Q}_k - \tilde{P}_k\tilde{B}_k R_k^{-1}\tilde{B}^\top_k\tilde{P}_k + \tilde{A}_k^\top\tilde{P}_k + \tilde{P}_k\tilde{A}_k.
    \end{split}
\end{equation}
Developing the terms, the following equations for $P_{11,k}$ and $P_{12,k}$ are obtained
\begin{align}
    &\rho P_{11,k}   \notag 
    = Q_k - P_{11,k}B_kR_k^{-1} B_k^\top P_{11,k} + A_k^\top P_{11,k} +  P_{11,k}A_k   \\[1ex]
    &\rho P_{12,k} 
    %
    = -Q_k H_k - P_{11,k} B_kR_k^{-1}B_k^\top P_{12,k} + A_k^\top P_{12,k} \notag\\ &\qquad+ P_{12,k}(\mathcal{A} - \mathcal{BR}^{-1}\mathcal{B}^\top(\mathcal{P}_{11}+\mathcal{P}_{12}))  \notag \\
    &\rho P_{22,k}
    = H_k^\top Q_k H_k - P_{12,k}^\top B_kR_k^{-1}B_k^\top P_{12,k} \notag\\&\qquad+\bar{G}^\top P_{22,k} + P_{22,k}\bar{G}.  \notag
\end{align}
Stacking the equations for $P_{11,k}$ for all $\mathcal{K}$ classes yields 
\begin{equation}
\label{eq:RDEP11}
    \rho \mathcal{P}_{11} = \mathcal{Q} + \mathcal{P}_{11}\mathcal{A} + \mathcal{A}^\top \mathcal{P}_{11} - \mathcal{P}_{11}\mathcal{BR}^{-1}\mathcal{B}^\top \mathcal{P}_{11}
\end{equation}
and stacking the equations for $P_{12,k}$ yields
\begin{align}
\label{eq:RDEP12}
    \rho \mathcal{P}_{12} & = -\mathcal{QH}+  (\mathcal{A}^\top-\mathcal{P}_{11}\mathcal{BR}^{-1}\mathcal{B}^\top)\mathcal{P}_{12} \notag \\&\hspace{-6pt}+ \mathcal{P}_{12}(\mathcal{A}-\mathcal{BR}^{-1}\mathcal{B}^\top \mathcal{P}_{11})- \mathcal{P}_{12}\mathcal{BR}^{-1}\mathcal{B}^\top \mathcal{P}_{12}
\end{align}
which are analogous to the equations in \eqref{eq:system_dynamics}, with $\mathcal{P}_{11} = \mathcal{P}$ and $\mathcal{P}_{12} = \Pi$.
Summing \eqref{eq:RDEP11} with \eqref{eq:RDEP12} and defining $\Omega = \mathcal{P}_{11}+\mathcal{P}_{12}$, an ARE for $\Omega$ is obtained
\begin{equation}\label{eq:OmegaAppendix}
    \rho\Omega = \mathcal{Q}(I_N-\mathcal{H}) + \Omega\mathcal{A} + \mathcal{A}^\top\Omega - \Omega\mathcal{BR}^{-1}\mathcal{B}^\top\Omega.
\end{equation}
Under Assumptions \ref{assum:A_B_stabi_A_Q_Obs} and \ref{assum:Hamiltonian_Q}, the pair ($\tilde{A}-\frac{1}{2}\rho I_{N}, \tilde{B}$) is stabilizable, and both \eqref{eq:RDEP11} and \eqref{eq:OmegaAppendix}, and thus also \eqref{eq:RDEP12}, all admit unique stabilizing solutions $\mathcal{P}_{11}, \Omega$ and $\mathcal{P}_{12}$ respectively, following \citep[Thm.~18]{Huang:2019} and Lemma~\ref{lem:HOmega_equivalency}. 
Therefore, the MFG strategy is uniquely given by  \eqref{eq:optimal_control_sol}-\eqref{eq:meanFieldDynamics}.     
This completes the proof.                                     

\section{Lemma use in the Proof of Prop.~\ref{prop:key_equations}} 
Let  $\bar{\mathcal{A}}\triangleq \mathcal{A}-\frac{1}{2}\rho I_N$ and $\mathcal{M}\triangleq\mathcal{BR}^{-1}\mathcal{B}^\top$. Consider the Hamiltonian matrix associated with the ARE in  \eqref{eq:dyn_pi}: 
\begin{equation}\label{eq:H-Pi}
    H_\Pi = \begin{bmatrix}
           \bar{\mathcal{A}}- \mathcal{M}\mathcal{P} & -\mathcal{M} \\
           \mathcal{QH} & -\bar{\mathcal{A}}^\top + \mathcal{PM}^\top
       \end{bmatrix}.
\end{equation}

\begin{lemma} \label{lem:HOmega_equivalency}
   Assume ($\bar{\mathcal{A}}, \mathcal{B}$) is stabilizable and the pair ($\bar{\mathcal{A}}, \mathcal{Q}$) is observable. Then the following hold:
   \begin{enumerate}
    \item $H_\Omega$ in \eqref{eq:matrixHomega} is strong $(N,N)$ c-splitting  if and only if $H_\Pi$ is strong $(N,N)$ c-splitting;
    \item  the $N$-dimensional stable invariant subspace   of $H_\Omega$ in \eqref{eq:matrixHomega} is a graph subspace if and only if the $N$-dimensional stable invariant subspace of $H_\Pi$ is a graph subspace.
\end{enumerate}
    
\end{lemma}
\begin{pf}

The Hamiltonian matrix in \eqref{eq:matrixHomega} is equivalently given by
\begin{equation} \label{eq:H-omega}
    H_{\Omega} = \begin{bmatrix}
         \bar{\mathcal{A}}& -\mathcal{M} \\
          -\mathcal{Q}(I_N-\mathcal{H}) & -\bar{\mathcal{A}}^{\top}
      \end{bmatrix}.
\end{equation}
Then
\begin{equation}\label{eq:H-relation2}
    \begin{bmatrix}
        I_N & 0\\
        -\mathcal{P} & I_N
    \end{bmatrix}
    H_\Omega
      \begin{bmatrix}
        I_N & 0\\
        \mathcal{P} & I_N
    \end{bmatrix} =
    \begin{bmatrix}
        \bar{\mathcal{A}}-\mathcal{MP} & -\mathcal{M}\\
        Z & -\bar{\mathcal{A}}^\top+\mathcal{PM}
    \end{bmatrix},
\end{equation}
with $Z \triangleq -\mathcal{P}\bar{\mathcal{A}}+\mathcal{PMP}-\mathcal{Q}(I_N-\mathcal{H}) - \bar{\mathcal{A}}^\top\mathcal{P}$. 
The stabilizability of ($\bar{\mathcal{A}}, \mathcal{B}$) and the observability of ($\bar{\mathcal{A}}, \mathcal{Q}$) ensure that~\eqref{eq:ARE_P} has a unique positive definite solution \citep[Thm.~4.1]{Wonham:1968}. 
Simplifying $Z$ with \eqref{eq:ARE_P} yields $Z = \mathcal{QH}$. The right hand side of \eqref{eq:H-relation2} becomes $H_\Pi$ in \eqref{eq:H-Pi}.
In addtion, we note that 
\begin{equation}
    \begin{bmatrix}
        I_N & 0 \\
        -\mathcal{P} & I_N
    \end{bmatrix}\begin{bmatrix}
        I_N & 0 \\
        \mathcal{P} & I_N
    \end{bmatrix} = I_{2N}.
\end{equation}
Hence, $H_\Omega$ and $H_\Pi$ are similar matrices. As a consequence, $H_\Omega$ and $H_\Pi$ share the same eigenvalues, and hence $H_\Omega$ is strong $(N,N)$ c-splitting if and only if $H_\Pi$ is strong $(N,N)$ c-splitting.
By the definition of stable graph subspace, it is easy to verify that
    the $N$-dimensional stable invariant subspace of $H_\Omega$ is a graph subspace if and only if the $N$-dimensional stable invariant subspace of $H_\Pi$ is a graph subspace.
\end{pf}

\end{document}